\documentclass{article}
\usepackage{spconf,amsmath,graphicx}
\usepackage{algorithm}
\usepackage{algorithmic}
\usepackage{subfigure}
\usepackage{hhline}  
\usepackage{amsmath,amssymb,amsfonts}
\usepackage{graphicx}
\usepackage{textcomp}
\usepackage{soul}
\usepackage{pifont}
\usepackage{xcolor}

\usepackage{multirow}
\usepackage{booktabs}
\usepackage{amsthm}
\newtheorem{theorem}{Theorem}[section]

\usepackage{amsfonts}

\title{MarkSweep: A No-box Removal Attack on AI-Generated Image Watermarking via Noise Intensification and Frequency-aware Denoising}
\name{Jie Cao  \qquad Zelin Zhang \qquad Qi Li \qquad Jianbing Ni}
  
  \address{Department of Electrical and Computer Engineering, Queen’s University, Kingston, Canada}

\begin{document}
\ninept
\maketitle
\begin{abstract}
AI watermarking embeds invisible signals within images to provide provenance information and identify content as AI-generated. In this paper, we introduce \textit{MarkSweep}, a novel watermark removal attack that effectively erases the embedded watermarks from AI-generated images without degrading visual quality. \textit{MarkSweep} first amplifies watermark noise in high-frequency regions via edge-aware Gaussian perturbations and injects it into clean images for training a denoising network. This network then integrates two modules, the learnable frequency decomposition module and the frequency-aware fusion module, to suppress amplified noise and eliminate watermark traces. Theoretical analysis and extensive experiments demonstrate that invisible watermarks are highly vulnerable to \textit{MarkSweep}, which effectively removes embedded watermarks, reducing the bit accuracy of HiDDeN and Stable Signature watermarking schemes to below 67\%, while preserving perceptual quality of AI-generated images.
\end{abstract}
\begin{keywords}
Image watermark, text-to-image model, watermark removal, copyright protection
\end{keywords}
\section{Introduction}
\label{sec:intro}

Generative Artificial Intelligence (Gen-AI) has emerged as a transformative paradigm capable of producing highly realistic and creative images. However, as the visual quality of these images continues to improve, distinguishing AI-generated content from human-created images has become increasingly difficult, raising significant risks such as misinformation and intellectual property violations. Invisible watermarking has gained prominence as a key technique for provenance verification and content attribution, addressing growing concerns over the misuse of AI-generated images~\cite{ren2024sok}. Traditional post-hoc watermarking methods, such as DwtDctSvd~\cite{navas2008dwt} and HiDDeN~\cite{zhu2018hidden}, can embed watermarks into AI-generated images after generation, while recently advanced in-generation watermarking can embed provenance signals directly during image generation~\cite{wen2023tree,fernandez2023stable,wang2024sleepermark}. These embedded watermarks can be detected with high bit accuracy, often exceeding 90\%, by the designated extractor during verification. 

Watermark removal attacks~\cite{zhao2024sok} have been proposed to erase embedded watermarks from AI-generated images, while preserving the perceptual fidelity of the host images, raising serious concerns about intellectual property violations and the potential misuse of generative models (see Fig.~\ref{fig:overview}). These attacks can be broadly categorized into three types: regeneration attacks~\cite{zhao2024invisible,an2024waves}, adversarial attacks~\cite{jiang2023evading,kassis2025unmarker}, and model-related attacks~\cite{hu2024stable,wang2024sleepermark}. 
Regeneration attacks reconstruct watermarked images using generative models such as variational autoencoders (VAEs) or diffusion models. Adversarial attacks introduce imperceptible perturbations to disrupt watermark detection, thereby increasing decoding errors while maintaining visual quality. Model-related attacks, in contrast, directly modify the parameters of generative models (e.g., diffusion decoders) so that all future outputs are watermark-free without altering individual images.
Although these attacks have demonstrated promising results, their effectiveness typically relies on restrictive assumptions: (i) access to a large number of watermarked images, often paired with clean counterparts; (ii) knowledge of full model parameters or high-fidelity surrogates; (iii) tolerance for significant image distortion; and (iv) dependence on complex architectures or iterative optimization, which slows down the attack process. These limitations highlight the need for a removal method with stronger generalization capabilities and greater removal effectiveness while preserving perceptual quality.

In this paper, we propose a new watermark removal attack, named \textit{MarkSweep}, to erase invisible image watermarks. Specifically, we first amplify watermark noise in high-frequency regions using edge-aware Gaussian perturbations, enabling effective denoising, and then design an end-to-end denoising network augmented with the learnable frequency decomposition module (LFDM) and the frequency-aware fusion module (FaFM) to accurately reconstruct watermark-free images. \textit{MarkSweep} is novel in that attackers do not need access to the designated extractors or paired clean counterparts. Instead, the removal is performed in no-box setting: the attacker requires only the target watermarked images to erase embedded signals. 
Finally, we demonstrate that \textit{MarkSweep} can effectively remove embedded watermarks for state-of-the-art (SOTA) schemes, reducing the bit accuracy of HiDDeN~\cite{zhu2018hidden} and Stable Signature~\cite{fernandez2023stable} to below 67\%, lower than their detection threshold, while preserving the perceptual quality of AI-generated images. Moreover, the time required to erase the watermark from a target image is under 1 second, significantly faster than the time of existing attacks, such as DiffusionAttack~\cite{zhao2024invisible} and UnMarker~\cite{kassis2025unmarker}.


\begin{figure}[t]
    \centering
    \includegraphics[width=\linewidth]{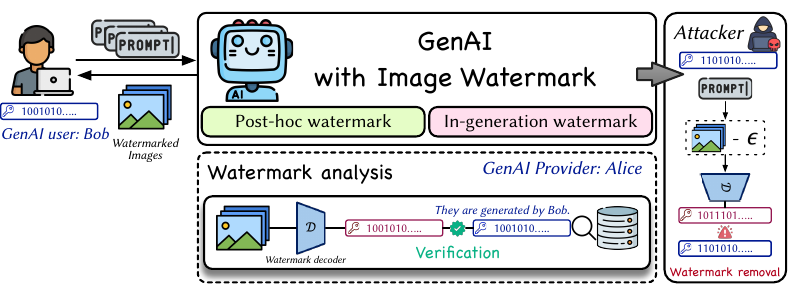}
    \vspace{-2em}
    \caption{Overview. Alice embeds watermarks in AI-generated images for attribution, but attackers can remove them, enabling the unauthorized dissemination of AI-generated content.}
    \label{fig:overview}
\vspace{-2em}
\end{figure}

\section{Background}

\subsection{Problem Formulation}
Let $x \in \mathbb{R}^{H \times W \times 3}$ denote a clean image, and $w \in \{0,1\}^k$ denote a binary watermark. A watermarking system consists of an encoder $E(x, w) = x_w$ that embeds $w$ into $x$, and a decoder $\mathcal{D}(x_w) = \hat{w}$ that extracts $w$. Watermark verification is successful when the bit accuracy (BA) (see Eq.~\ref{eq:bit acc}) exceeds a threshold chosen to achieve a target False Positive Rate (FPR). Th threshold $\tau$ is predefined so that the FPR remains below an acceptable level (e.g., $10^{-6}$), preventing clean images from being falsely detected as watermarked~\cite{fernandez2023stable,yang2024gaussian}.
\vspace{-1em}
\begin{equation}\label{eq:bit acc}
\text{BA}(w, \hat{w}) = \frac{1}{m} \sum_{i=1}^{m} \mathbb{I}[w_i = \hat{w}_i].
\vspace{-2em}
\end{equation}

\subsection{Threat Model}
\textbf{Adversary’s goals.} The adversary aims to craft a perturbation $\epsilon$ such that the perturbed image $\tilde{x} = x_w - \epsilon$ evades watermark detection, while keeping $\tilde{x}$ visually close to both the clean image $x$ and the watermarked image $x_w$.

\textbf{Adversary’s knowledge and capability.} 
We adopt the no-box setting of~\cite{kassis2025unmarker}, in which the adversary has no access to the internals of the generative model 
$\mathcal{G}$, the watermark encoder $\mathcal{E}$, or the decoder $\mathcal{D}$. The adversary may only leverage public and unwatermarked datasets as well as has limited short-term GPU access, which precludes training a surrogate generative model.

\section{METHODOLOGY}

\subsection{Overview}
As illustrated in Fig.~\ref{fig:method_pipeline}, \textit{MarkSweep} removes invisible watermarks through a noise–intensification–denoising pipeline. Given a watermarked image \(x_w\), we first apply noise intensification to emphasize high-frequency watermark patterns, producing \(x_n\). A pre-trained ResNet-50 encoder~\cite{he2016deep} then extracts visual features, which are decomposed into multi-frequency representations by the LFDM. The FaFM further refines and fuses these features, suppressing amplified noise. Finally, a decoder reconstructs the clean image \(\hat{x}\), followed by a frozen RealESRGAN~\cite{wang2021real} module to enhance image details and overall perceptual quality.

\begin{figure}[h]
    \centering
    \includegraphics[width=\linewidth]{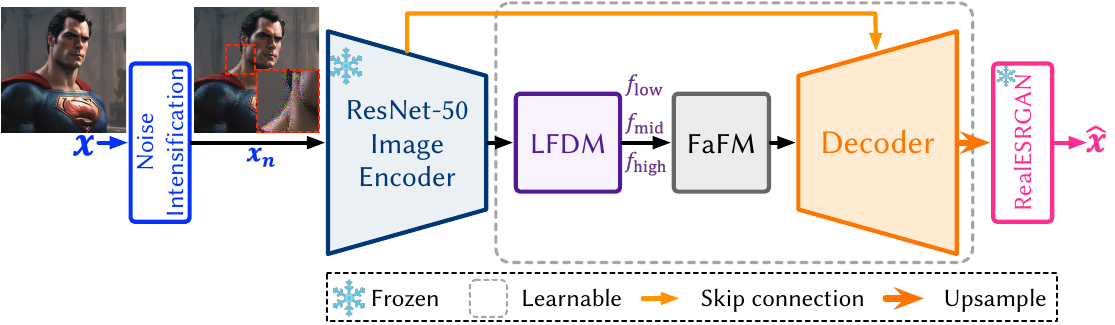}
        \vspace{-1em}
        \caption{An illustration of MarkSweep pipeline.}
    \label{fig:method_pipeline}
            \vspace{-2em}
\end{figure}

\subsection{Intensification of Watermark Noise}

In invisible image watermarking, the embedding process is typically confined to high-frequency regions~\cite{kassis2025unmarker,wu2025there} (e.g., edges and textures), to which the human visual system is less sensitive. This strategy allows the watermark to remain imperceptible while preserving visual quality. Building on this observation, we define the pixel-wise residual as $\Delta_{w} = x - x_{w},$ treating the watermark as additive noise. An effective removal attack then aims to estimate a perturbation \(\epsilon\) that closely approximates the watermark signal \(\Delta_{w}\):
\vspace{-0.5em}
\begin{equation}\label{eq:diff}
\epsilon^* = \arg\min_{\epsilon} \| \epsilon - \Delta_w  \|_p.
\vspace{-0.5em}
\end{equation}

However, our empirical analysis shows that the residual noise has very low magnitude, approximately $\|\Delta_w \| \sim \mathcal{N}(\mu,\sigma^2)$ with $\mu \in [2,10]$, 
which makes direct removal challenging even for SOTA denoising networks when watermarked images are unavailable.

To overcome this limitation, we introduce a noise intensification method defined as
$x_n = n + x$, where $n = \mathcal{A}(x,\mu,\sigma,s)$, which amplifies high-frequency components to make watermark patterns more distinguishable. Specifically, we inject non-uniform Gaussian noise $n\sim\mathcal{N}(\mu,\sigma^{2})$ into edge regions detected by the Canny operator, and apply morphological dilation (structuring element size $s$) to distinguish core and extended edges. A gradient-weighted mask then imposes stronger perturbations on core edges and weaker perturbations on surrounding regions, thereby approximating the distribution of watermark noise. This noise-enhanced image \(x_n\) enables more effective watermark removal in the subsequent denoising stage.
\vspace{-1em}

\subsection{Denoising Network} 
As shown in Fig.~\ref{fig:method_pipeline}, in this stage, we propose a denoising network $\mathcal{M}_{\theta}$, where $\theta$ denotes the trainable parameters. The network adopts an encoder–decoder architecture to remove amplified watermark noise. We first employ a ResNet-50 pre-trained on ImageNet as the backbone encoder $\mathcal{E}(\cdot)$, which extracts feature maps $f = \mathcal{E}(x_{n}) \in \mathbb{R}^{c\times h \times w}$ from the noise-enhanced input $x_{n}$.

\textbf{LFDM.} Both Gaussian and watermark noise predominantly reside in high-frequency components, offering critical spectral cues for effective removal. We integrate a LFDM into the denoising network, mapping features from the spatial to the frequency domain. A learnable band-pass filter decomposes signals into low-, mid-, and high-frequency bands, enabling adaptive processing for each band and stronger suppression of noise in high-frequency regions. 

Given an input feature map $f$, we first transform it into the frequency domain feature $\mathbf{F}$ using the 2D Fast Fourier Transform (FFT) $\mathcal{F}$. Then, we introduce learnable thresholds $\boldsymbol{\gamma} = \{\gamma_0,\gamma_1\}$ to define the complete frequency band boundaries $\mathbf{b} = \{0, \gamma_0, \gamma_1,1\}$.

For band-pass filtering, the three frequency bands are processed using the respective masks defined as:
\vspace{-0.5em}
\begin{equation}
 M_i = \phi\big(k(R - b_{i})\big) \odot \phi\big(k(b_{i+1} - R)\big), b_i\in\mathbf{b},
     \vspace{-0.5em}
\end{equation}
where $\phi$ denotes the Sigmoid function, $R$ is the normalized radial frequency in the Fourier domain; $k = 10$ controls transition sharpness between bands, and $\odot$ indicates element-wise multiplication. Each band response is computed as $f_i = \mathcal{F}^{-1}(\mathbf{F} \odot M_i) \cdot w_i$, in which 
$w_{i}$ is a learnable weight and $\mathcal{F}^{-1}$ is the inverse FFT. Finally, the LFDM outputs the band-specific features $\{f_{\text{low}}, f_{\text{mid}}, f_{\text{high}}\}$.

\textbf{FaFM.} To effectively integrate features from different frequency bands, we propose a novel FaFM. This module is designed to intelligently aggregate three frequency bands in a hierarchical and adaptive manner. FaFM comprises a joint attention mechanism and a final feature fusion stage.

FaFM first employs a joint attention mechanism that simultaneously models inter-band channel dependencies and shared spatial saliency within a unified representation space. Specifically, the band features are concatenated along the channel dimension $\tilde{f} = [f_{\text{low}}; f_{\text{mid}}; f_{\text{high}}]$, with $\tilde{f} \in \mathbb{R}^{3c\times h \times w}$, and a channel-attention function $\mathcal{A}_{c}(\cdot)$ is applied to produce the attention weights $W_c$:
\vspace{-0.5em}
\begin{equation}
    W_c = \phi(\mathrm{Conv}_{1\times1}(\delta(\mathrm{Conv}_{1\times1}(\mathrm{GAP}(\tilde{f}))))),
\vspace{-0.5em}
\end{equation}
where $\phi$ denotes ReLU function.

Channel-wise band averages are concatenated and processed by spatial attention $\mathcal{A}_{s}(\cdot)$ to obtain $W_s$:
\vspace{-0.5em}
\begin{equation}
    W_s = \phi\!\left(\mathrm{Conv}_{7\times7}\left([\mathrm{Avg}_{ch}(f_i)]_{i \in \{\text{low}, \text{mid}, \text{high}\}}\right)\right).
\vspace{-0.5em}
\end{equation}

We apply the joint attention mechanism and incorporate a channel attention weight $W_c \in \mathbb{R}^{c \times 3 \times 1 \times 1}$ to capture frequency-specific channel dependencies, along with a spatial attention weight $W_s \in \mathbb{R}^{1\times1\times h \times w}$ to emphasize the important spatial regions within each frequency band. 

Finally, the weighted features are obtained via element-wise multiplication and then fused by averaging over the frequency-band dimension to produce the final representation $f^{*}\in \mathbb{R}^{c \times h \times w}$:
\begin{equation}
    f^{*} = \operatorname{Avg}(\operatorname{Stack}[f_{\text{low}}; f_{\text{mid}}; f_{\text{high}}]\odot W_c \odot W_s).
    \vspace{-0.5em}
\end{equation}

\textbf{Decoder.} To reconstruct the final denoised image, we employ a learnable decoder that is structurally symmetric to the ResNet-based encoder. Following the conventional UNet's decoder structure~\cite{ronneberger2015u}, the decoder comprises a sequence of upsampling modules. Each module doubles the feature map resolution via bilinear interpolation, applies a convolution to adjust the number of channels, and uses a skip connection to concatenate the upsampled features with the corresponding high-resolution encoder feature map. This preserves fine spatial details that may be lost during upsampling. Through this process, the decoder progressively transforms the deepest fused representation $f^{*}$ back to the original input resolution, ultimately generating the final denoised image $\hat{x}$.

\textbf{Objective Function.} To ensure that the denoised images retain both visual similarity and perceptual quality with respect to the clean images, we train $\mathcal{M}_{\theta}$ by minimizing the following loss function:
\begin{equation}
\begin{aligned}
  \mathcal{L}(\theta)&= \mathbb{E}_{x_n, x}\big[ 
\lambda_1 \mathcal{L}_{\text{LPIPS}}(\hat{x}, x) 
+ \lambda_2 \mathcal{L}_{\text{MSE}}(\hat{x}, x) \\&+ \lambda_3\mathcal{L}_\text{FFT}(\hat{x}, x)
 + \lambda_4 \mathcal{L}_{\text{MSE}}(x_n-\hat{x}, n)
\big] ,
\end{aligned}
    \vspace{-0.5em}
\end{equation}
where $\mathcal{L}_{\text{MSE}}$ and $\mathcal{L}_{\text{LPIPS}}$ denote the mean squared error and Learned Perceptual Image Patch Similarity (LPIPS) loss~\cite{zhang2018unreasonable} between the denoised image $\hat{x}$ and the clean image $x$, measuring pixel-level and perceptual similarity. To complement the pixel-domain loss, we employ $\mathcal{L}_\text{FFT}$ to constrain the consistency of the amplitude spectrum in the 2D frequency domain. In addition, $\mathcal{L}_{\text{MSE}}(x_n-\hat{x}, n)$ is used to supervise the estimation of the noise added to the clean image. 

Additionally, since our method intensifies and denoises noise primarily in high-frequency edge regions, slight blurring may occur.  
To mitigate this, we apply a pre-trained Real-ESRGAN model~\cite{wang2021real} for same-size super-resolution (SR) as shown in Fig.~\ref{fig:method_pipeline}, enhancing perceptual detail and texture sharpness.

  

\section{Theoretical Analysis}
In this section, we theoretically demonstrate that \textit{MarkSweep} is capable of effectively removing invisible image watermarks. 

\begin{theorem}[Information-Theoretic Guarantee of Watermark Removal] \label{th:1}
Let $x$ denote a clean image and $w$ be a binary watermark. For the Markov chain 
\begin{equation}
w \rightarrow x_w \rightarrow x_w+n \rightarrow \hat{x}=\mathcal{M}_\theta(x_w+n),
    \vspace{-0.5em}
\end{equation}
with noise intensification $n=\mathcal{A}(x_w,\mu,\sigma,s)$ and denoising network $\mathcal{M}_\theta$, the mutual information satisfies:
\begin{equation}
\boxed{I(w;\hat{x}) \leq I(w;x_w+n) \leq I(w;x_w)}.
\end{equation}
\end{theorem}
\vspace{-0.5em}

\noindent\textbf{\textit{Proof.}} 
Let $y = x_w + n$ with $n=\mathcal{A}(x_w,\mu,\sigma,s, r)$, where $r \perp (w,x_w)$ denotes exogenous randomness independent of $(w,x_w)$.
We assume that the joint probability density function of all random variables involved in the Markov chain can be factorized as:
\begin{equation}
p(w,x_w,y,\hat{x})=p(w)\,p(x_w|w)\,p(y\,|\,x_w)\,p(\hat{x}\,|\,y),
\end{equation}
\vspace{-0.5em}
According to the data processing inequality~\cite{beaudry2011intuitive}, we obtain:
\begin{equation}
I(w;\hat{x}) \leq I(w;y) \leq I(w;x_w).
    \vspace{-0.5em}
\end{equation}

In watermarking systems, $w$ can be decoded from $x_w$ with high accuracy, so $I(w;x_w)\approx H(w)$~\cite{mittelholzer1999information}. Noise intensification and denoising weaken this dependency, yielding 
\begin{equation}
    I(w;\hat{x}) < I(w;x_w).
\end{equation}


\begin{theorem}[Attack Effectiveness]\label{th:2}
Let $\hat{x}$ be a denoised image. If $I(w;\hat{x}) < I(w;x_w)$, then
\begin{equation}
    \boxed{BA(w,\mathcal{D}(\hat{x})) < BA(w,\mathcal{D}(x_w))}.
\end{equation}
i.e., the watermark becomes unrecoverable from $\hat{x}$.
\end{theorem}
\vspace{-1em}

\noindent\textbf{\textit{Proof.}} By Theorem~\ref{th:1}, reducing $I(w;\hat{x})$ leads to an increase in $H(w|\hat{x})$. Applying Fano’s inequality~\cite{mackay2003information}, we obtain:
\begin{equation}
    H(w|\hat{x}) \leq H(P_e) + P_e \log(|\mathcal{W}|-1),
        \vspace{-0.5em}
\end{equation}
where $H(P_e) = -P_e \log P_e - (1 - P_e) \log(1 - P_e)$ is the binary entropy function. $ P_e = \Pr(\hat{w} \neq w)$ denotes the decoding error probability of the $\hat{w}$, where $\hat{w}$ is the extracted watermark by $\mathcal{D}(\hat{x})$. $|\mathcal{W}|$ denotes the size of the watermark space. Rearranging gives the lower bound:
\begin{equation}
    P_e \geq \frac{H(w|\hat{x})-1}{\log(|\mathcal{W}|-1)}.
    \vspace{-0.5em}
\end{equation}
Thus, as $H(w|\hat{x})$ increases, $P_e$ increases, resulting in a decrease in the expected bit accuracy:
\begin{equation}
    BA(w,\mathcal{D}(\hat{x})) < BA(w,\mathcal{D}(x_w)).
    \vspace{-0.5em}
\end{equation}

In summary, the above proof establishes a formal guarantee that as the attack increases conditional entropy more effectively, the corresponding error rate in watermark decoding inevitably rises.

\begin{table*}[t]
    \footnotesize
    \caption{Visual quality comparison of various watermark removal attacks.}
    \centering

    \resizebox{\textwidth}{!}{
    \begin{tabular}{ c cccc cccc cccc cccc}
        \toprule
        \multirow{2}{*}{\textbf{WM}}
        & \multicolumn{4}{c}{\textbf{SS}} 
        & \multicolumn{4}{c}{\textbf{Yu }}
        & \multicolumn{4}{c}{\textbf{PTW }} 
        & \multicolumn{4}{c}{\textbf{HiDDeN }}
        \\ 
        
        \cmidrule(lr){2-5} 
        \cmidrule(lr){6-9} 
        \cmidrule(lr){10-13} 
        \cmidrule(lr){14-17} 

        &PSNR &SSIM &LPIPS &A-FINE
        &PSNR &SSIM &LPIPS &A-FINE
        &PSNR &SSIM &LPIPS &A-FINE
        &PSNR &SSIM &LPIPS &A-FINE

        \\
        \midrule
        JPEG 
        &41.61 &0.98&0.02& 37.93
        &47.80&0.99&0.01&44.96
        &37.19&0.98&0.04&36.51
        &39.68 &	0.97 &0.02 &29.42	
        \\
        
        Crop 
        &14.27&0.46&0.43&49.96
        &18.43&0.53&0.22&51.16
        &15.95&0.47&0.34&53.35
        &16.40&0.39&0.39&52.06

        \\
        
        Blur 
        &31.60&0.92&0.17&52.66
        &35.01&0.95&0.10&54.75

        &33.68&0.94&0.13&59.66
        &29.64&0.88&0.21&62.82

        \\
        
        Noise 
        &26.22&0.53&0.33&54.85
        &26.13&0.55&0.33&60.71
        &26.18&0.54&0.34&55.89
        &26.06&0.58&0.32&55.32
        \\
        \midrule
        DiffusionAttack 
        &27.01&0.78&0.32&40.47
        &31.96&0.87&0.21&37.29
        &26.01&0.73&0.27&34.81
        &30.22&0.84&0.18&38.77
        \\

        VAEAttack 
        &32.12&0.89&0.25&38.03
        &33.38&0.91&0.21&67.51
        &32.30&0.89&0.25&57.39
        &29.93&0.84&0.31&49.71
        \\
        
        UnMarker
        &21.86&0.52&0.01&   45.99
        &17.77&0.42&0.14& 53.78
        &15.84&0.43&0.13& 44.36
        &23.54&0.62&0.21&48.07
        \\
        \midrule
        \textit{MarkSweep} \textbf{w/o} SR
        &30.89  &0.89 &0.13 &41.89
        &21.73  &0.65&0.19&53.92
        &31.54  &0.91   &0.11 &45.88
        &28.93  &0.84 &0.20&45.46
        \\
        \textit{MarkSweep}
        &28.69&0.86&0.20&39.56
        &29.75&0.88&0.17&50.96
        &28.99&0.87&0.17&43.55
        &27.85&0.82&0.23&50.67
     \\
     \bottomrule   
    \end{tabular}}
    \label{tab:metric}
    \vspace{-1em}
\end{table*}

\begin{figure*}[t]
    \centering
    \includegraphics[width=0.6\linewidth]{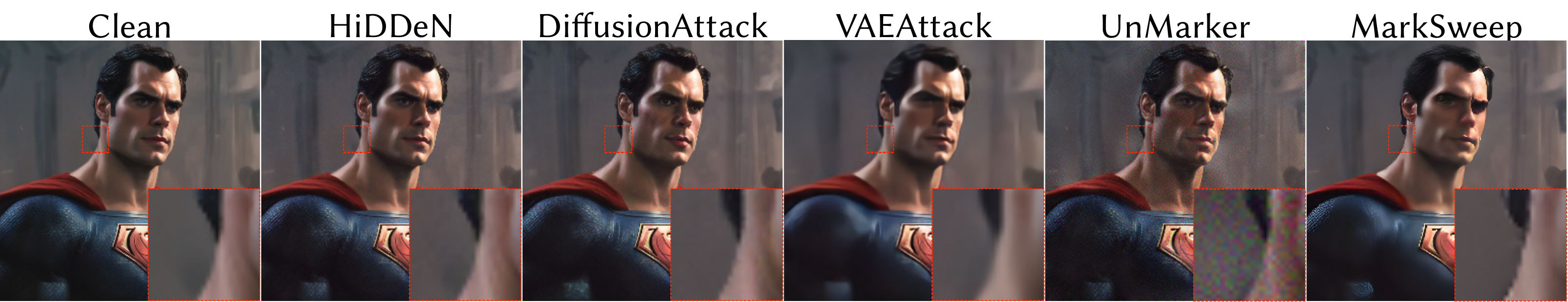}
        \vspace{-1em}
    \caption{Outputs of the SOTA removal attacks against HiDDeN watermarking method.}
    \label{fig:HiDDeNSample}
    \vspace{-2em}
\end{figure*}

\section{Experiment}
\subsection{Experimental Setup}

\textbf{Training.} Our denoising network is trained on 5,000 clean images from the natural dataset MS-COCO~\cite{lin2014microsoft} and 5,000 images from the AI-generated dataset GenImage~\cite{zhu2023genimage}, enabling the model to learn from both natural and AI-generated image patterns. The implementation is based on PyTorch, and training is conducted for 500 epochs with a batch size of 80 and an initial learning rate of 0.001. We use the Adam optimizer~\cite{adam} for parameter updates, and a warm-up cosine learning rate scheduler~\cite{loshchilov2016sgdr} is adopted to ensure stable convergence. All experiments are performed on a single NVIDIA A6000 GPU. During training and evaluation, image patches of size $224 \times 224$ are used, with 80\% of the data allocated for training and 20\% for validation. For noise intensification, the parameters are set as $\mu=0$, $\sigma=50$, and $s=5$. The loss function weights are configured as $\lambda_1 = 1$, $\lambda_2 = 35$, $\lambda_3 = 0.2$, and $\lambda_4 = 20$.

\textbf{Watermarking methods.} We compare \textit{MarkSweep} with SS~\cite{fernandez2023stable}, Yu~\cite{yuresponsible}, PTW~\cite{lukas2023ptw}, and HiDDeN~\cite{zhu2018hidden}, using detection thresholds of 69\%, 61\%, 70\%, and 73\%, respectively, as in~\cite{kassis2025unmarker}. Yu and PTW are generator-level watermarking methods designed for Generative Adversarial Networks (GANs), SS is a text-to-image watermarking scheme based on Stable Diffusion Models (SDMs), and HiDDeN represents a post-hoc watermarking approach. For evaluation, we use CelebA-HQ as the test dataset for Yu and PTW, Stable Diffusion Prompts (SDP) for SS, and images generated from SDP with SDXL 1.0 (base) for HiDDeN. For each method, 100 images are randomly sampled for attack evaluation.


\textbf{Metrics.} We evaluate the quality of de-watermarked image using Peak Signal-to-Noise Ratio (PSNR)\textcolor{green}{$\blacktriangle$}, Structural Similarity Index Measure (SSIM\textcolor{green}{$\blacktriangle$}), LPIPS\textcolor{red}{$\blacktriangledown$}, and the scaled Adaptive Fidelity-Naturalness Evaluator (A-FINE\textcolor{red}{$\blacktriangledown$})~\cite{chen2025toward}. A-FINE is a novel FR-IQA metric that relaxes the need for a perfect reference, allowing reliable quality assessment even when the reference is imperfect or inferior to the test image. This makes it well suited for evaluating de-watermarked images when watermark-free references are unavailable. In addition, we use BA (Eq. \eqref{eq:bit acc}) to quantify the effectiveness of watermark removal, and attack speed (AS) to evaluate its computational efficiency.

\textbf{Baseline attacks.} The baseline attack comprises a sequence of standard image distortions: Gaussian blur (kernel size = 10), the addition of Gaussian noise ($\mu = 0$, $\sigma = 0.05$), a 90\% central crop, and JPEG compression.

\textbf{SOTA attacks.} We compare \textit{MarkSweep} against SOTA watermark removal attacks, including DiffusionAttack~\cite{zhao2024invisible}, VAEAttack~\cite{an2024waves}, and UnMarker~\cite{kassis2025unmarker}. All of these methods are representative black-box or no-box attacks with demonstrated effectiveness.

\vspace{-1em}


    
        

\subsection{Attack Performance}
\textbf{Attack Effectiveness.} As shown in Table~\ref{tab:BA}, \textit{MarkSweep} shows promising attack effectiveness. It not only reduces BA but also pushes it below the detection thresholds of several watermarking schemes. For instance, on HiDDeN, it achieves 51.32\% BA, below the 73\% threshold and lower than VAEAttack (52.20\%), while remaining close to UnMarker (50.33\%). Overall, \textit{MarkSweep} poses a significant threat for invisible watermarking techniques. 

\begin{table}[t]
    \footnotesize
    \caption{BA comparison under various attacks.}
    \centering

    \resizebox{0.48\textwidth}{!}{
    \begin{tabular}{ c cc cc cc c|c }  
        \toprule
       \textbf{ Attacks} & \textbf{JPEG} &\textbf{Crop} & \textbf{Blur} & \textbf{Noise} & \textbf{DiffusionAttack} &\textbf{VAEAttack} &\textbf{UnMarker} & \textbf{\textit{MarkSweep}}\\
        \midrule
        SS &98.10 & 86.87 &98.22   &91.61 &48.40    &65.33&56.38&66.83\\
        Yu &89.75 &74.47 &94.57         &64.58&58.34&56.94&59.95&59.24\\
        PTW &83.29 &85.59 &75.20        &50.76&57.59&66.38&63.24&72.19\\
        HiDDeN &65.00 &88.73 &64.18     &51.41&51.63&52.20&50.33&51.32\\
     \bottomrule   
    \end{tabular}}
    \label{tab:BA}
    \vspace{-2em}
\end{table}

\textbf{Visual Quality.} \textit{MarkSweep} achieves competitive visual quality compared to other watermark removal attacks across all watermarking schemes. As shown in Table~\ref{tab:metric}, in most cases, \textit{MarkSweep} maintains high PSNR (e.g., 28.99 dB on PTW and 28.69 dB on SS) and SSIM scores ($\geq 0.82$), indicating that the structural fidelity of the images is well preserved. Although some classical distortion attacks, such as JPEG compression, yield higher PSNR and SSIM values, they often fail to fully remove the watermark. \textit{MarkSweep} effectively balances removal effectiveness with perceptual quality, as reflected in its moderate LPIPS scores and consistent A-FINE performance. Although \textit{MarkSweep} alone maintains image quality, the SR module provides additional benefits in selected scenarios by improving the visual quality of denoised images. These results demonstrate that \textit{MarkSweep} provides a favorable trade-off between watermark removal success and image quality preservation.

Furthermore, as shown in Fig.~\ref{fig:HiDDeNSample}, the effectiveness of \textit{MarkSweep} is obvious: It removes the high-frequency watermark signals introduced by HiDDeN while avoiding the excessive blurring observed in DiffusionAttack and VAEAttack. Although UnMarker achieves strong attack performance and excellent image quality metrics, it introduces noticeable checkerboard artifacts that are easily perceptible to the human visual system.

\begin{figure}[t]
    \centering
    \includegraphics[width=0.8\linewidth]{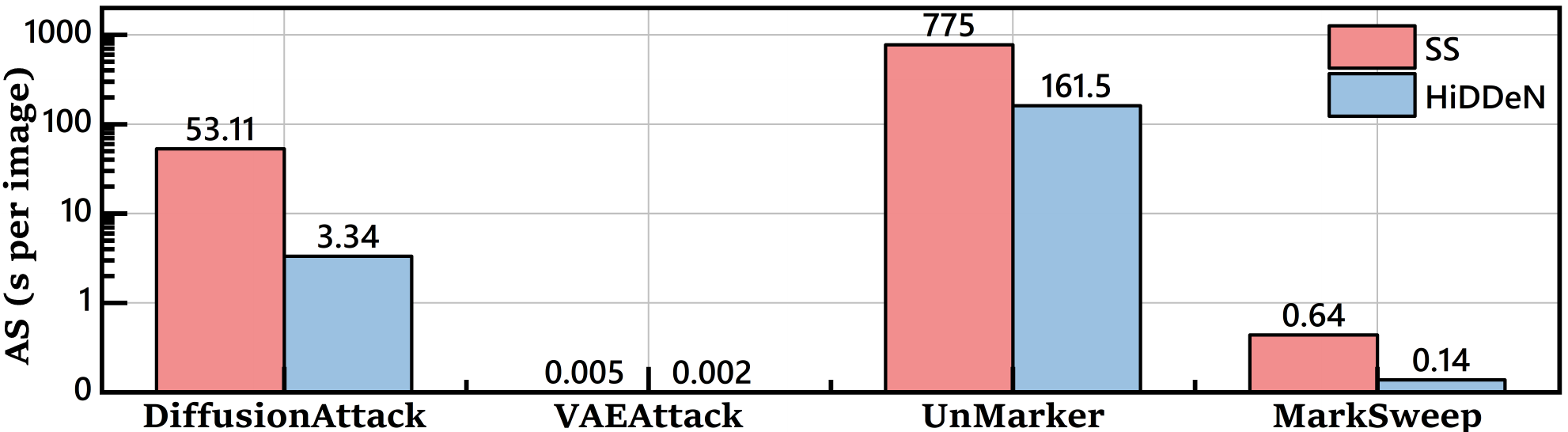}
        \vspace{-1em}
    \caption{Comparison of AS across SOTA attacks.}
    \label{fig:AS}
    \vspace{-2em}
\end{figure}

\textbf{Attack Efficiency.} Fig.~\ref{fig:AS} shows that \textit{MarkSweep} is the second fastest attack, only slower than VAEAttack. \textit{MarkSweep} requires 0.64 seconds per image on SS and 0.14 seconds on HiDDeN, which is 83$\times$ faster than DiffusionAttack on SS and 24$\times$ faster on HiDDeN. Compared to UnMarker, \textit{MarkSweep} is over 1,210$\times$ faster on SS and more than 1,153$\times$ faster on HiDDeN. Unlike \textit{MarkSweep}, which maintains stable efficiency, UnMarker relies on costly iterative optimization, resulting in extremely slow attack speeds.

\textbf{Limitation.} Although our MarkSweep can effectively remove the watermarks and preserve image quality for emerging fine-tuning-based and post-hoc watermarking schemes (e.g., StableSignature, HiDDeN), the effectiveness is limited on the emerging initial noise-based watermarking schemes (e.g., Gaussian Shading \cite{yang2024gaussian}), aligning with SOTA attacks (e.g., UnMarker~\cite{kassis2025unmarker}).

\section{Conclusion}
In this paper, we propose \textit{MarkSweep}, a novel AI-generated image watermark removal attack, which effectively suppresses invisible watermarks while maintaining high perceptual quality, based on noise intensification and suppression. \textit{MarkSweep} operates under no-box settings, making it a realistic and substantial threat to current watermarking schemes. Therefore, we encourage new designs of more robust watermarking schemes for AI-generated images against frequency-aware attacks, thereby mitigating the growing risks of disinformation and copyright infringement involving Gen-AI.
Our theoretical analysis provides a provable attack guarantee, and extensive experiments demonstrate that \textit{MarkSweep} achieves a trade-off between attack effectiveness, efficiency, and image fidelity, outperforming existing watermark removal methods. In future work, we will evaluate the robustness of emerging text-to-image watermarking methods that are based on initial noise of SDMs.

{\small
\bibliographystyle{IEEEbib}
\bibliography{refs}
}
\end{document}